\documentclass[twocolumn,showpacs,aps,prl,letter,amsmath,amssymb,superscriptaddress]{revtex4-1}
\usepackage{bm,color,bbm}
\usepackage{mathtools,graphicx,natbib}

\usepackage{multibib}
\newcites{xxx}{Supplemental References}

\usepackage[english]{babel}
\usepackage[utf8x]{inputenc}
\usepackage[T1]{fontenc}
\usepackage{lmodern}
\usepackage{float}

\usepackage{siunitx} 

\newcommand{\beq}{\begin{equation}}
\newcommand{\eeq}{\end{equation}}
\newcommand{\bqa}{\begin{eqnarray}}
\newcommand{\eqa}{\end{eqnarray}}

\newcommand{\bra}[1]{ \langle{#1} |}
\newcommand{\ket}[1]{ |{#1} \rangle}

\newcommand{\mbb}[1]{\mathbb{#1}}
\newcommand{\II}{\mbb{I}}
\newcommand{\CC}{\mbb{C}}
\newcommand{\Tr}{\operatorname{Tr}}
\newcommand{\ketbra}[2]{\left| #1 \right\rangle \left\langle #2 \right|}
\newcommand{\braket}[2]{\left\langle #1 \vert #2 \right\rangle}
\newcommand{\abs}[1]{\left| #1\right|}
\newcommand{\norm}[1]{\left\| #1 \right\|}
\usepackage{amsthm}
\newtheorem{lemma}{Lemma}
\newtheorem*{lemma*}{Lemma}

\newtheorem*{corollary*}{Corollary}
\theoremstyle{remark}
\newtheorem{remark}{Remark}

\definecolor{nblue}{rgb}{0.3,0.32,0.85}
\definecolor{nblue}{rgb}{0.0,0.0,0.0}
\newcommand{\blu}{\color{nblue}}
\definecolor{maroon}{rgb}{0.7,0,0}

\definecolor{ngreen}{rgb}{0.3,0.7,0.3}

\definecolor{golden}{rgb}{0.8,0.6,0.1}

\usepackage{amsmath}
\usepackage{graphicx}
\usepackage{epstopdf}

\begin{document}

\title{Conclusive Experimental Demonstration of One-Way Einstein-Podolsky-Rosen Steering}

\author{Nora Tischler}
\affiliation{Centre for Quantum Computation and Communication Technology (Australian Research Council), Centre for Quantum Dynamics, Griffith University, Brisbane, QLD 4111, Australia.}

\author{Farzad Ghafari}
\affiliation{Centre for Quantum Computation and Communication Technology (Australian Research Council), Centre for Quantum Dynamics, Griffith University, Brisbane, QLD 4111, Australia.}

\author{Travis J. Baker}
\affiliation{Centre for Quantum Computation and Communication Technology (Australian Research Council), Centre for Quantum Dynamics, Griffith University, Brisbane, QLD 4111, Australia.}

\author{Sergei Slussarenko}
\affiliation{Centre for Quantum Computation and Communication Technology (Australian Research Council), Centre for Quantum Dynamics, Griffith University, Brisbane, QLD 4111, Australia.}

\author{Raj B. Patel}
\affiliation{Centre for Quantum Computation and Communication Technology (Australian Research Council), Centre for Quantum Dynamics, Griffith University, Brisbane, QLD 4111, Australia.}
\affiliation{Centre for Quantum Computation and Communication Technology  (Australian Research Council), Quantum Photonics Laboratory, School of Engineering, RMIT University, Melbourne, Victoria 3000, Australia.}

\author{Morgan M. Weston}
\affiliation{Centre for Quantum Computation and Communication Technology (Australian Research Council), Centre for Quantum Dynamics, Griffith University, Brisbane, QLD 4111, Australia.}

\author{Sabine  Wollmann}
\affiliation{Centre for Quantum Computation and Communication Technology (Australian Research Council), Centre for Quantum Dynamics, Griffith University, Brisbane, QLD 4111, Australia.}
\affiliation{Quantum Engineering Technology Labs, H. H. Wills Physics Laboratory and Department of Electrical \& Electronic Engineering, University of Bristol, Bristol BS8 1FD, UK.
}

\author{Lynden K. Shalm}
\affiliation{National Institute of Standards and Technology, 325 Broadway, Boulder, Colorado 80305, USA.}

\author{Varun B. Verma}
\affiliation{National Institute of Standards and Technology, 325 Broadway, Boulder, Colorado 80305, USA.}

\author{Sae Woo Nam}
\affiliation{National Institute of Standards and Technology, 325 Broadway, Boulder, Colorado 80305, USA.}

\author{{\blu H. Chau Nguyen}}
\affiliation{{\blu Naturwissenschaftlich-Technische Fakult\"{a}t, Universit\"{a}t Siegen, Walter-Flex-Straße 3, D-57068 Siegen, Germany.}}

\author{Howard M. Wiseman}
\affiliation{Centre for Quantum Computation and Communication Technology (Australian Research Council), Centre for Quantum Dynamics, Griffith University, Brisbane, QLD 4111, Australia.}

\author{Geoff J. Pryde}
\email{g.pryde@griffith.edu.au}
\affiliation{Centre for Quantum Computation and Communication Technology (Australian Research Council), Centre for Quantum Dynamics, Griffith University, Brisbane, QLD 4111, Australia.}

\begin{abstract}
Einstein-Podolsky-Rosen steering is a quantum phenomenon wherein one party influences, or steers, the state of a distant party's particle beyond what could be achieved with a separable state, by making measurements on one half of an entangled state. This type of quantum nonlocality stands out through its asymmetric setting, and even allows for cases where one party can steer the other, but where the reverse is not true. A series of experiments have demonstrated one-way steering in the past, but all were based on significant limiting assumptions. These consisted either of restrictions on the type of allowed measurements, or of assumptions about the quantum state at hand, by mapping to a specific family of states and analyzing the ideal target state rather than the real experimental state. Here, we present the first experimental demonstration of one-way steering free of such assumptions. {\blu We achieve this using a new sufficient condition for nonsteerability, and,} although not required by our analysis, using a novel source of extremely high-quality photonic Werner states.

\end{abstract}

\maketitle

\paragraph{Introduction.---} 

One of the most noteworthy and fundamental features of quantum mechanics is the fact that it admits stronger correlations between distant objects than what would be possible in a classical world. Quantum correlations can be categorized into the following classes, which form a strict hierarchy \cite{Wiseman2007,Quintino2015, Saunders2010}: entanglement is a superset of Einstein-Podolsky-Rosen (EPR) steerability, which in turn is a superset of Bell nonlocality. Out of these, steering is special in that it allows for, and in fact intrinsically contains, asymmetry. Steering is operationally defined as a quantum information task, where one untrusted party (for instance called Alice) tries to convince another distant, trusted party (Bob) that they share entanglement. Bob asks Alice to make certain measurements on her quantum system (e.g., particle) and to announce the measurement outcomes, but is not sure whether Alice answers honestly, or indeed, even has a particle. He also makes corresponding measurements on his particle and checks whether the correlations of their measurement outcomes rule out a so-called local hidden state model for his particle, thereby proving shared entanglement \cite{Wiseman2007}. 

Interestingly, the steering task allows for the case of one-way steerable states, for which steering is possible in one direction but impossible in the reverse direction \cite{Bowles2014}. One-way steering is of foundational interest, since it is a striking manifestation of asymmetry that does not exist for entanglement and Bell-nonlocality. It also has applications in device-independent quantum key distribution \cite{Branciard2012}. To observe one-way steering, one needs to demonstrate steering in one direction, by violating a steering inequality. In addition, one must establish that it would be impossible to achieve steering in the opposite direction. Our scheme, which allows for arbitrary measurements and rigorously takes into account losses and the real experimental quantum state, is illustrated in Fig.\ \ref{fig1}.

\begin{figure}[b]
	\centering
	\includegraphics[width=8.2cm]{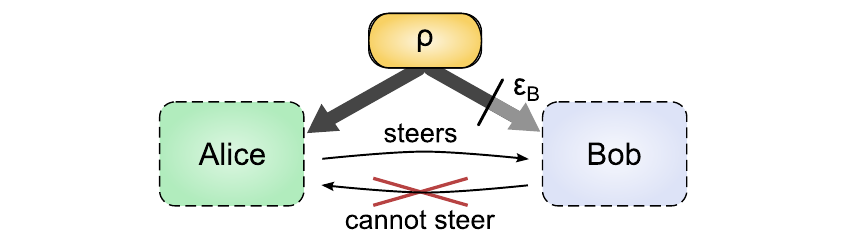}

	\caption{Scheme for demonstrating one-way steering. A two-qubit quantum state is distributed to Alice and Bob, with a lossy channel on the way to Bob, such that his probability of obtaining his qubit is $\varepsilon_B$. A detection-loophole-free steering test demonstrates that Alice can steer Bob's state. At the same time, it is established that Bob cannot steer Alice's state for any choice of measurements, based directly on the reconstructed experimental quantum state $\rho$ and the measured efficiency $\varepsilon_B$.}
	\label{fig1}
\end{figure}

Since the question whether one-way steering is possible was first raised in the seminal paper of Ref.\ \cite{Wiseman2007}, considerable progress has been made on the topic \cite{Midgley2010,Handchen2012,Bowles2014,Skrzypczyk2014,Evans2014,Quintino2015,Wollmann2016,Bowles2016,Sun2016,Rao2016,Xiao2017,Olsen2017,Wang2017,Baker2018}. First, the original question was answered in the affirmative, and this gave rise to the quest to fully understand and demonstrate the phenomenon. An overarching effort of these works has been the elimination of assumptions. 

\begin{figure*}[t]
	\centering
	\includegraphics[width=17.2cm]{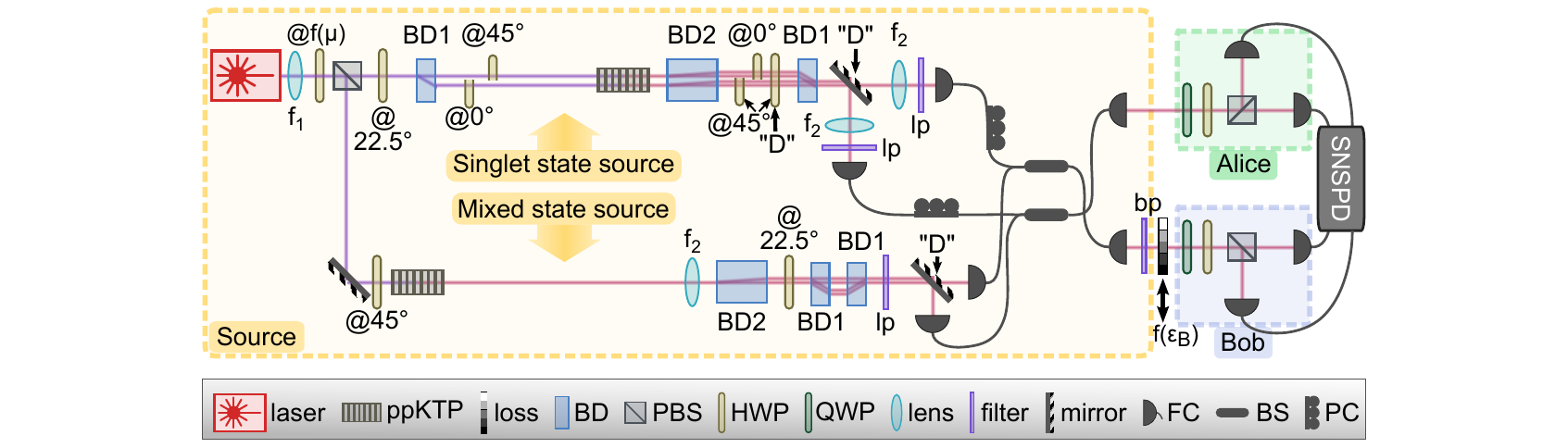}

	\caption{Experimental setup. As detailed in the main text, the tunable source of telecom-wavelength two-qubit Werner states is constructed as an incoherent superposition of the outputs from a singlet state source and a source of the maximally mixed two-qubit state. After a variable loss in one arm, polarization measurements are carried out in Alice and Bob's stations, enabling quantum state reconstruction and steering tests. Abbreviations: ppKTP, periodically poled potassium titanyl phosphate; BD, (polarizing) beam displacer; PBS, polarizing beam splitter; HWP, half-wave plate; QWP, quarter-wave plate; FC, {\blu (single-mode)} fiber coupler;  BS, (fiber) beam splitter; PC, (fiber) polarization controller; SNSPD, superconducting nanowire single-photon detectors; lp, longpass; bp, bandpass; "D", D-shaped element with a horizontal cut that is not apparent from the top view. BD2 implements a vertical beam displacement (see main text), which is illustrated in the diagram by the slightly separated pairs of beams. For further details about the experimental elements, see the {\blu SM} \cite[Sec.~III]{SI}.}
	\label{fig2}
\end{figure*}

On the theory side, the ultimate, so far unattained, goal would be to establish practical necessary and sufficient conditions for the steerability of arbitrary quantum states using arbitrary measurements, which are described by positive operator-valued measures (POVMs). Examples of one-way steerable states have been identified assuming projective measurements \cite{Bowles2014,Evans2014,Bowles2016} and {\blu for POVMs \cite{Skrzypczyk2014,Quintino2015,Wollmann2016}}. While specific example states provide conclusive proof that one-way steering is possible in principle, they are challenging to work with in real experiments. Real states in the laboratory generally deviate from the ideal target states, so the ability to account for these deviations is crucial. A practical, sufficient condition for the nonsteerability of arbitrary two-qubit states under the assumption of projective measurements is known \cite{Bowles2016}. Recently, a practical, sufficient condition for the nonsteerability of generic two-qubit states with loss was also established for {\blu restricted projective measurements} \cite{Baker2018} {\blu(see explanation in the Supplemental Material (SM) \cite[Sec.~I]{SI})}. \nocite{Hirsch2013}
\nocite{Barrett2002}

The elimination of assumptions has also been a key development on the experimental side. Several experiments relied on assumptions about the measurements. The first demonstration of one-way steering was restricted to the case of Gaussian measurements \cite{Handchen2012}. This was followed by a demonstration that was restricted to two-setting projective measurements \cite{Sun2016}, and another one assuming multi-setting projective measurements \cite{Xiao2017}. In contrast to these post-selection-based experiments, an experiment by some of us and co-workers had no detection loophole and therefore, its analysis could take into account vacuum state contributions to the quantum state \cite{Wollmann2016}. Also, unlike the previous experiments, it made no assumptions about the measurement. It did, however, make an assumption about the type of quantum state; an analysis for Werner states was applied to the experimentally achieved state, which exhibited a high fidelity with a Werner state. However, drawing conclusions based on high fidelities can be problematic in general \cite{Peters2004}, and caution is also warranted for the case at hand \cite{Baker2018}\cite[Sec.~V]{SI}.

\begin{table*}[!t]
    \label{tab:table1}
    \begin{tabular}{c|c|c|c|c|c|c} 
      $\mu$ & 0.9978$\pm$0.0003  & 0.797$\pm$0.001 &  0.603$\pm$0.001& 0.398$\pm$0.002 & 0.198$\pm$0.002 & 0.007$\pm$0.002  \\
\hline
      state fidelity& 0.9981$\pm$0.0002  & 0.9964$\pm$0.0004 &  0.9983$\pm$0.0002& 0.9985$\pm$0.0001 & 0.9986$\pm$0.0001 & 0.9983$\pm$0.0001
    \end{tabular}
 \caption{Tunability and quality of the experimental quantum state. For six different pump conditions, we determine the Werner state $\rho_{\mathrm{w}}$ with which the experimental state $\rho$ has the highest fidelity. Listed are the parameter $\mu$ of the closest Werner state, and the corresponding fidelity, defined as $\left[\mathrm{Tr}\left(\sqrt{\sqrt{\rho}\rho_{\mathrm{w}}\sqrt{\rho}} \right)\right]^2$. Uncertainties are obtained from Monte Carlo simulations based on Poisson distributed counts to generate two hundred variations on each of the experimental tomography measurement results. The reconstructed density matrices are shown in the {\blu SM} \cite[Sec.~II]{SI}.}
\end{table*}

Here, we present the first fully rigorous experimental demonstration of two-qubit one-way steering. In one direction, we demonstrate the violation of a steering inequality with the detection loophole closed. Using the recent theory result of Ref.\ \cite{Baker2018} {\blu and new theory developed in the SM \cite[Sec.~I]{SI}, we provide} a sufficient condition for nonsteerability, valid for general POVMs {\blu performed on} arbitrary two-qubit states with loss, {\blu and} conclusively show that our state is not steerable in the opposite direction. We further demonstrate the impact of different experimental parameters, which highlights the delicate nature of experimental one-way steering. Although the formalism does not assume it, our experimental states are very close to two-qubit Werner states. Two-qubit Werner states comprise a one-parameter family of states written as $\rho_{\mathrm{w}}=\mu |\Psi^-\rangle\langle\Psi^-|+\left(1-\mu \right)/4 \mathbf{I}_4$, where $|\Psi^- \rangle=\left(|01\rangle - |10\rangle \right)/\sqrt{2}$ is the singlet state and $\mathbf{I}_4$ is the $4\times 4$ identity matrix. These states represent a well-known example of mixed states \cite{Werner1989}, {\blu with their} purity determined by the Werner state parameter $\mu\in \left[ 0,1\right]$.

A number of sources of photonic two-qubit Werner states have been reported in the past \cite{Zhang2002,Barbieri2004,Liu2017,Sun2018,Li2018}. Here, we use a new type of photon source, producing high quality states that have unprecedented fidelities with Werner states. 

\paragraph{Werner state source.---}

Our photonic source of Werner states is based on spontaneous parametric down-conversion (SPDC) with a picosecond pulsed pump laser, producing photon pairs at telecom wavelength, with the quantum state encoded in the polarization degree of freedom ($|H\rangle\equiv|0\rangle$, |$V\rangle\equiv|1\rangle$). It is constructed as an incoherent superposition of a singlet state source and a source of maximally mixed photon pairs. Our design provides high heralding efficiencies and full control of the Werner state parameter $\mu$.

The detailed setup is illustrated in Fig.\ \ref{fig2}. A 775 nm pulsed laser with variable power and a pulse length of 1 ps acts as the pump for the two individual sources comprising the overall source. After passing through a focusing lens, the pump beam is divided between the two sources with a controllable splitting ratio by using a half-wave plate (HWP) and polarizing beam splitter (PBS). 

The singlet state source is based on the design of Ref.\ \cite{Shalm2015} and essentially implements a superposition of two SPDC events within a beam displacer interferometer. The pump passes through a HWP that sets its polarization to an equal superposition of horizontal (H) and vertical (V) components, which are then horizontally split into two beams by the first beam displacer (BD). The next two HWPs act to make the polarizations of both beams H, appropriate for the subsequent down-conversion process, while matching the path lengths of the two beams. The beams then pump the 15 mm long periodically poled potassium titanyl phosphate (ppKTP) crystal in two places, enabling degenerate type-II SPDC. The second BD separates signal and idler photons vertically, resulting in a total of four down-converted photon beams for the one photon pair. The next three HWPs modify the polarizations of the beams such that the left two beams are H polarized, while the right two beams are V polarized. This allows overlapping the signal photon from the two different down-conversion beams with the third BD, and likewise for the idler photon. A D-shaped mirror separates the propagation directions of the signal and idler photon beams, each of which are collimated, have the pump light filtered out with a longpass filter, and are {\blu coupled into single-mode fiber. To transform the maximally entangled state of $|HH\rangle$ and $|VV\rangle$ to one of $|HV\rangle$ and $|VH\rangle$, a} $90^{\circ}$ polarization rotation for one of the two photons is implemented with in-fibre polarization controllers, and the phase $\phi$ of the target state $\left(|HV\rangle -e^{i\phi}|VH\rangle\right)/\sqrt{2}$ can be controlled through slight tilting of the first BD, or by adjusting the crystal temperature.

The design of the mixed state source is such that a separable photon pair is created, and then each photon is fully depolarized, yielding the target state $ \mathbf{I}_4/4$. The pump beam passes through a ppKTP crystal identical to the one in the entangled state source, creating one H and one V polarized photon, which are collimated with a lens. The two photons are vertically separated into two beams with a BD, and subsequently their polarization is rotated by $45^{\circ}$ with a HWP. An imbalanced BD interferometer, in which one polarization component passes straight through and the other component undergoes spatial walk-off twice (in opposite directions), decoheres the polarization of each of the signal and idler photon completely. A longpass filter discards the pump light, before the propagation directions of the signal and idler beams are separated with a D-shaped mirror and they are fiber coupled.

The two individual sources are mixed using 50:50 fiber beam splitters, which combine the signal photon contributions coming from the two sources, and likewise for the idler photon. This mixing is incoherent, since the path lengths through the two sources are sufficiently different. Finally, a bandpass filter in Bob's arm narrows the biphoton spectrum \cite[Sec.~III]{SI}, which enhances the polarization state quality for the singlet source. By tuning the relative power of the pump in the two individual sources, the parameter $\mu$ can be controlled. For a range of relative power values, we perform quantum state tomography of the photon pairs using a combined pump power setting of $\sim 75$ mW, and determine the fidelities with the closest Werner states, as detailed in Table I. These fidelities are the highest reported values to date. 

A further noteworthy feature of our source is its high heralding efficiency. Despite a $50\%$ loss due to the mixing of the two individual sources via 50:50 beam splitters and the additional components in the measurement apparatus, we still obtain typical heralding efficiencies (defined as detected coincidences divided by the detected singles of the opposite arm, {\blu also called Klyshko efficiency \cite{Klyshko1980}}) of $0.3100\pm 0.0003$ and $0.2345\pm0.0002$, for the arm without and with the bandpass filter, respectively. The high heralding efficiencies are made possible by the choices of the pump beam waist, the detection beam waist, and high-efficiency superconducting nanowire single-photon detectors \cite{Marsili2013}\cite[Sec.~III]{SI}. 

\paragraph{One-way steering.---}

To demonstrate one-way steering, we use the same setup as before, with some minor modifications. To add controllable loss, we insert a multi-setting neutral density filter before the detection apparatus in Bob's arm, which lowers his overall heralding efficiency to $\varepsilon_B$. We also increase the total pump power to $\sim 300$ mW in order to maintain a sufficiently high signal-to-noise ratio with the attenuated beam against the detector dark counts, which are $\sim100$ per second \cite[Sec.~III]{SI}. 
As shown in Fig.\ \ref{fig3} and explained below, the output of our Werner state source together with the added loss creates one-way steerable states, provided that the values of $\mu$ and $\varepsilon_B$ are suitably chosen. Note that in the steering experiment, some of the fidelities with the closest Werner states are lower than the results shown in Table I, but our subsequent analysis is robust as it makes no assumption of the experimental states being Werner states.

To demonstrate one-way steering, we perform two sets of measurements. The purpose of the first set is to show steering from Alice to Bob. This is done via a steering test with $n=6$ measurement settings, using a platonic-solid measurement scheme \cite{Bennet2012}. Detection-loophole-free steering is demonstrated if the correlations of the measurement outcomes are sufficiently large, resulting in a steering parameter that exceeds the $n=6$ steering bound (the definition of the steering parameter is provided in the {\blu SM} \cite[Sec.~IV]{SI}). The bound is a function of Alice's heralding efficiency, $\varepsilon_A$, because in this task, she is the person who is attempting to steer her opponent's state. {\blu Our experiment thus necessarily closes the detection efficiency loophole, 
though we make no claim to close the space-like-separation loophole.}

The purpose of the second set of measurements is to establish nonsteerability from Bob to Alice, for general POVMs. This is achieved via a quantum state tomography, through which our experimental density matrix is reconstructed. Based on the density matrix and Bob's experimentally-measured heralding efficiency, we test the criterion for nonsteerability {\blu derived in the SM \cite[Sec.~I]{SI}}, 
\begin{equation} 
N_{\rm POVM}\le 1.
\label{Traveqn}
\end{equation}
Here $N_{\rm POVM}$ is defined as
\begin{equation} 
N_{\rm POVM}=\max_{\hat{\bf{x}}\in \widehat{\mathbb{R}^3}} [(1-{\blu 3} \varepsilon_B) |{\bf{b}}\cdot\hat{{\bf{x}} }|+ \frac{{\blu 3} \varepsilon_B}{{\blu 2}}  (1+({\bf{b}}\cdot \hat{{\bf{x}}})^2)+||T\hat{{\bf{x}}}||],
\label{Traveqn2}
\end{equation}
where $\bf{b}$ is Bob's local Bloch vector, $T$ is the correlation matrix of the quantum state in its canonical form, and $||...||$ denotes the 2-norm. The maximization is carried out over all unit vectors $\hat{{\bf{x}}}$ in three dimensions.
{\blu This criterion is stronger than that in Ref.\ \cite{Baker2018}, and its derivation (see the SM \cite[Sec.~I]{SI}) is more rigorous: It ensures nonsteerability from Bob to Alice without restricting Bob's measurements to POVMs on the photonic qubit subspace.}

\begin{figure}[b!]
	\includegraphics[width=0.5\textwidth]{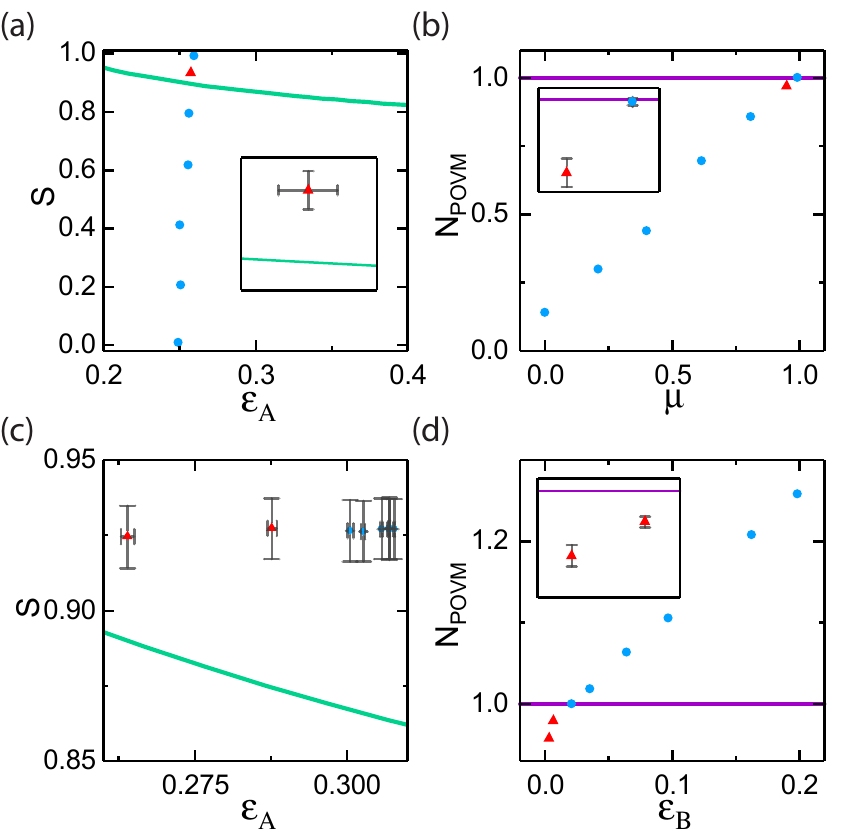}

	\caption{Experimental demonstration of one-way steering. The upper (lower) panels contain results for a set of states where the effective $\mu$ parameter (Bob's heralding efficiency $\varepsilon_B$) is varied. The left panels contain results for a detection-loophole-free steering test for Alice to steer Bob, where the steering parameter $S$ for $n=6$ measurement settings is plotted against Alice's heralding efficiency $\varepsilon_A$. For any data point above the bound given by the green line, steering from Alice to Bob is demonstrated. The right panels depict results for {\blu our} sufficient condition for nonsteerability with arbitrary POVMs, and data points below the purple line are conclusively nonsteerable from Bob to Alice. Each data point from one of the left panels corresponds to a point from the right panel, with a pair representing a specific quantum state. Data points in order of increasing $S$ from (a) correspond to the points with increasing $\mu$ in (b). Similarly, data points in order of increasing $\varepsilon_A$ in (c) correspond to those with increasing $\varepsilon_B$ in (d). Uncertainties for the steering parameter are calculated as $\Delta S=\sqrt{\Delta S(\mathrm{systematic})^2+\Delta S(\mathrm{statistical})^2}$ \cite{Bennet2012}. The other uncertainties are based on Monte Carlo simulations of the measurement outcomes, using two hundred samples of Poisson distributed counts. Where uncertainties are small and error bars would reduce the clarity of the plots, the error bars are not shown. However, the relevant uncertainties are provided in the insets, which enlarge areas of interest. Conclusively one-way steerable states are marked by the red triangles.}
	\label{fig3}
\end{figure}

Obtaining a steering parameter in one direction above the steering bound and showing, based on the density matrix and heralding efficiency, that the corresponding quantum state is unsteerable in the opposite direction, successfully demonstrates one-way steering.
We perform the measurements for two sets of quantum states. In the first set, we keep the loss added by the neutral density filter fixed {\blu such that Bob's heralding efficiency is $\varepsilon_B=(2.52\pm0.03)\times 10^{-3}$}, while varying $\mu$. The results from the steering test are shown in Fig.\ \ref{fig3}(a), and the results from the corresponding test of the sufficient condition for nonsteerability in the opposite direction are depicted in Fig.\ \ref{fig3}(b). Of all the $\mu$ values shown, only one, marked by the red triangle, is conclusively one-way steerable (steering bound violation from Alice to Bob by 3.8 standard deviations (s.d.), and fulfillment of the sufficient condition for nonsteerability from Bob to Alice with a margin of {\blu 5.3} s.d.). 

For the second set of quantum states, we keep $\mu$ fixed {\blu at $0.951\pm0.004$}, while varying the loss added by the neutral density filter. The results of the steering test are given in Fig.\ \ref{fig3}(c), and the results from the corresponding test of the sufficient condition for nonsteerability in the opposite direction are shown in Fig.\ \ref{fig3}(d). Here, the two {\blu states} corresponding to the lowest $\varepsilon_B$ values are conclusively one-way steerable (steering bound violation by 3.3 and 5.2 s.d., and nonsteerability with margins of {\blu 6.0 and 5.6} s.d., respectively). The {\blu states with} higher $\varepsilon_B$ are no longer conclusively nonsteerable from Bob to Alice. The ability to further reduce $\varepsilon_B$ is limited for technical reasons only, namely the decreasing signal-to-noise ratio due to dark counts, which reduces the ideally constant measured heralding efficiency $\varepsilon_A$ when the attenuation is very high.

The results highlight that in practice, demonstrating one-way steering based on two-qubit states with loss requires a balance between (i) having sufficient correlations to observe steering in one direction, while (ii) keeping the loss needed to conclude nonsteerability in the opposite direction at a technically feasible level.

\paragraph{Discussion.---}

Our experiment is based on a two-qubit state with loss. Implementing loss in a quantum information protocol is relatively straightforward. In fact, some amount of loss is generally unavoidable in practice, so even if the loss was not actively leveraged, an experimental analysis would need to account for it in any case. Therefore, a two-qubit state with loss  is well-motivated from a practical perspective.  

It is worth emphasizing that we establish nonsteerability by checking against {\blu our \emph{sufficient}} condition for nonsteerability {\blu (Eq.\ (\ref{Traveqn}))}. This condition offers the best currently available method for demonstrating the nonsteerability of general two-qubit states with loss, allowing for general POVMs. However, the condition is not proven to be tight, so it is possible that tighter conditions will be found in the future. For example, it might be possible to show that the necessary and sufficient conditions for steerability coincide for projective measurements and POVMs, which would then make it easier to demonstrate one-way steering.  However, the fact that we work with a sufficient condition for nonsteerability means that our results are conclusive now, and will remain so, even in the event that tighter conditions are found in the future.

\paragraph{Conclusion.---}

In this work, we present a new, high-heralding-efficiency photon-pair source that produces quantum states with very large fidelities with two-qubit Werner states, and provides full control of the Werner state parameter. {\blu We use the source and a new sufficient condition for nonsteerability to achieve} a rigorous demonstration of two-qubit one-way steering free of previous limiting assumptions about the experimental quantum state or measurement.


\begin{acknowledgements}
This work was supported by the ARC Centre of Excellence CE110001027. F.G., T.J.B., and S.W. acknowledge financial support through Australian Government Research Training Program Scholarships. {\blu C.N. acknowledges support by the DFG and the ERC (Consolidator Grant 683107/TempoQ).}
\end{acknowledgements}


\bibliography{biblio_OWS2}


\onecolumngrid

\newpage 


\setcounter{equation}{0}
\renewcommand{\theequation}{S.\arabic{equation}}

\setcounter{figure}{0}
\renewcommand{\thefigure}{S\arabic{figure}}


\renewcommand{\thepage}{Supplemental Material -- \arabic{page}/5}
\setcounter{page}{1}

\begin{center}
  \large {\bf Supplemental Material: Conclusive Experimental Demonstration of One-Way Einstein-Podolsky-Rosen Steering}
\end{center}
\vskip 1em
\begin{center}
   \lineskip .75em%
  \begin{tabular}[t]{c}
    Nora Tischler, Farzad Ghafari, Travis J. Baker, Sergei Slussarenko, Raj B. Patel\\
    Morgan M. Weston, Sabine  Wollmann, Lynden K. Shalm, Varun B. Verma\\
    Sae Woo Nam, H. Chau Nguyen, Howard M. Wiseman, Geoff J. Pryde
  \end{tabular}
\end{center}

{\blu
\section{I. The sufficient condition for nonsteerability}

In this section, we provide a proof for the nonsteerability criterion {\blu of} Eqs.\ (1) and (2) in the main text. For theoretical convenience, we consider quantum steering from Alice to Bob as in Ref.~\cite{Baker2018}; a simple permutation {\blu at} the end allows us to obtain the {\blu criterion} for {\blu nonsteerability} from Bob to Alice as described by {\blu Eqs.} (1) and (2) in the main text. In the following, we use PVMs to denote \emph{projective measurements}, and $n$-POVMs to denote \emph{positive operator-valued measures of $n$ outcomes}. In fact, we work very often with the notion of $2$-POVMs with a rank-$1$ projection component. Those $2$-POVMs are of the form $(Q,\II-Q)$ with $\II$ being the identity operator and $Q$ being a rank-$1$ projection. 

We start with repeating the proof of Lemma 1 of Quintino et al.\ in Ref.~\cite{Quintino2015} in its more general form (see Ref.~\citexxx{Hirsch2013}).

\begin{lemma}[Quintino et al.]
If a state $\rho$ of $\CC^d\times \CC^{d'}$ is nonsteerable for $2$-POVMs with a rank-$1$ projection component, then the state
\begin{equation}
\tilde{\rho}= \frac{1}{d} \rho + \frac{d-1}{d} \sigma_A \otimes \rho_B
\end{equation}
with arbitrary state $\sigma_A$ and $\rho_B= \Tr_A[\rho]$ is nonsteerable with arbitrary POVMs.
\label{lem:noisy_POVM}
\end{lemma}
\begin{proof}
Since an {\blu extremal} POVM  $E=(E_1,E_2,\ldots,E_n)$ has at most $n=d^2$ nonzero components, we can fix $n=d^2$. Moreover, we can assume that the components of the POVM $E$ are rank-$1$, namely $E_i=\alpha_i Q_i$ for some rank-$1$ projections $Q_i$ and $0 \le \alpha_i \le 1$ (since all other POVMs can be post-processed from these; see, e.g., Ref.~\citexxx{Barrett2002}).  In proving this lemma, it is convenient to rewrite a measurement $E$ in the direct sum form $E=\oplus_{i=1}^{n} E_i$. The steering ensemble{\blu---the set of Bob's reduced states conditioned on $E$---} is {\blu thus} written as $\oplus_{i=1}^{n} \Tr_A[\tilde{\rho} E_i \otimes \II_B]$. 
Now we claim the following identity:
\begin{equation}
\bigoplus_{k=1}^{n} \Tr_A[\tilde{\rho} E_k \otimes \II_B] = \sum_{i=1}^{n} \sum_{j=1}^{n} \frac{\alpha_i \beta_j}{d} \bigoplus_{k=1}^{n} \Tr_A \left[ \rho (\delta_{ik} Q_i + \delta_{jk} (\II_A-Q_i)) \otimes \II_B  \right],
\label{eq:noise_identity}
\end{equation}
where $\beta_j= \Tr (\sigma_A E_j)$. Note that $\sum_{i=1}^{n} \frac{\alpha_i}{d}=\sum_{j=1}^{n} \beta_j=1$. This identity can be proved straightforwardly by passing the sums over the direct sum and {\blu performing} them explicitly. This identity tells {\blu us} that the steering ensemble $\oplus_{k=1}^{n} \Tr_A[\tilde{\rho} E_k \otimes \II_B]$ can be written as a convex combination of $n^2$ steering ensembles  $\oplus_{k=1}^{n} \Tr_A [ \rho (\delta_{ik} Q_i + \delta_{jk} (\II_A-Q_i)) \otimes \II_B]$, each with probability $\frac{\alpha_i \beta_j}{d}$. Note then that the latter ensembles $\oplus_{k=1}^{n} \Tr [ \rho (\delta_{ik} Q_i + \delta_{jk} (\II_A-Q_i)) \otimes \II_B]$ correspond to steering $\rho$ with $2$-POVMs with a rank-$1$ projection component, $(Q_i, \II_A - Q_i)$, when {\blu empty components} are discarded.  These ensembles can all be locally simulated from a local hidden state (LHS) ensemble by assumption. It follows that the former steering ensemble $\oplus_{k=1}^{n} \Tr_A[\tilde{\rho} E_k \otimes \II_B]$ can also be locally simulated from the LHS ensemble. In other words, $\tilde{\rho}$ is nonsteerable with $n$-POVMs.
\end{proof}

\begin{remark}
We can translate this mathematical proof to a physical one. Alice's aim is to simulate steering of $\tilde{\rho}$ with POVMs $E=(\alpha_1 Q_1,\alpha_2 Q_2,\ldots,\alpha_n Q_n)$ (chosen by Bob). Alice provides Bob with the LHS ensemble that she can use to simulate steering of $\rho$ with $2$-POVMs with one rank-$1$ projection component. With probability $\alpha_i\beta_j/d$, she chooses a pair $(i,j)$. She then simulates the outcomes $i$ and $j$ as if they are outcomes of $Q_i$ and $\II_A-Q_i$ in the measurement $(Q_i,\II_A-Q_i)$, respectively. This is slightly different from the protocol in Ref.~\citexxx{Hirsch2013}, but the result is the same. 
\end{remark}

Suppose we have a two-qubit state $\rho$ acting on $\CC^2 \otimes \CC^2$. A {\blu state} with loss can be described by a state on $\CC^3 \otimes \CC^2$ as
\begin{equation}
\rho_{\varepsilon_A}= \varepsilon_A \rho + (1-\varepsilon_A) \ketbra{\nu}{\nu} \otimes \rho_B, 
\label{eq:two_qubit_loss}
\end{equation}
where $\ket{\nu}$ is the vacuum state, which is orthogonal to the standard qubit states $\ket{0}$ and $\ket{1}$ (see the main text and  Ref.~\cite{Baker2018}). Here {\blu $\varepsilon_A$} is Alice's heralding efficiency, $0\le \varepsilon_A \le 1$.  

In Ref.~\cite{Baker2018}, some of the present authors considered in particular the steerability of $\rho_{\varepsilon_A}$ with PVMs which are restricted to the form $(\ketbra{\varphi_1}{\varphi_1},\ketbra{\varphi_2}{\varphi_2},\ketbra{\nu}{\nu})$, where $\ket{\varphi_1}$ and $\ket{\varphi_2}$ are two orthogonal qubit states (which are both orthogonal to the vacuum $\ket{\nu}$). These PVMs are not the most general PVMs, and will be referred to as \emph{restricted} PVMs in the following. Let $\rho$ be written as 
\begin{equation}
\rho= \frac{1}{4} \left[\II \otimes \II + \sum_{i=1}^{3}  a_i \sigma_i \otimes \II + \sum_{i,j=1}^{3} T_{ij } \sigma_i \otimes \sigma_j \right],  
\end{equation}
where $\mathbf{a}=(a_1,a_2,a_3)^T$ is Alice's Bloch vector, and $T$ is the correlation matrix, which can be assumed to be diagonal. 
Theorem~1 in Ref.~\cite{Baker2018} states that ${\blu \rho_{\varepsilon_A}}$ is nonsteerable with restricted PVMs if
\begin{equation}
\max_{\hat{\bf{x}}} \left[ (1-\varepsilon_A) \abs{\bf{a} \cdot \hat{\bf{x}}} + \frac{\varepsilon_A}{2} (1+\abs{\bf{a} \cdot \hat{\bf{x}}}^2) + \norm{T\hat{\bf{x}}}\right] \le 1,
\label{eq:restricted_PVMs}
\end{equation}
where the maximization is taken over all unit vectors $\hat{\bf{x}}$ of $\mathbb{R}^3$.

We now show that nonsteerability with restricted PVMs implies nonsteerability with all $2$-POVMs with a rank-$1$ projection component. This then allows us to apply Lemma~\ref{lem:noisy_POVM} to construct a state which is nonsteerable with arbitrary POVMs.

\begin{lemma}
Consider the {\blu two-qubit state} with loss $\rho_{\varepsilon_A}$ in~\eqref{eq:two_qubit_loss}. If $\rho_{\varepsilon_A}$ is nonsteerable with restricted PVMs, then it is nonsteerable with all $2$-POVMs with a rank-$1$ projection component.
\label{lem:two_qubit_loss_2POVMs}
\end{lemma}

\begin{proof}
Consider Alice making a measurement of the form $(\ketbra{\psi}{\psi},\II_A-\ketbra{\psi}{\psi})$, where $\II_A$ is the identity operator on Alice's space $\CC^3$, $\ket{\psi}$ is a state of $\CC^3$, which may {\blu be nonorthogonal} to the vacuum  $\ket{\nu}$. It is sufficient to show that the steering outcome corresponding to $\ketbra{\psi}{\psi}$ can be locally simulated. We start with finding the steering outcome $\Tr_A [\rho_{\varepsilon_A} (\ketbra{\psi}{\psi} \otimes \II_B) ]$ for $\ketbra{\psi}{\psi}$, which is
\begin{equation}
\varepsilon_A \Tr_A (\rho \ketbra{\psi}{\psi} \otimes \II_B) + (1-\varepsilon_A) \abs{\braket{\nu}{\psi}}^2 \rho_B,
\end{equation}  
where $\II_B$ is {\blu the} identity operator acting on Bob's space. 
Since $\Tr_B [\rho]$ has no support on $\ketbra{\nu}{\nu}$, we can insert the projection $\Pi=\II_A-\ketbra{\nu}{\nu}$, which projects $\CC^3$ to the two-qubit space $\CC^2$, such that $\rho= (\Pi \otimes \II_B) \rho (\Pi \otimes \II_B)$. Therefore, 
\begin{equation}
\Tr_A (\rho \ketbra{\psi}{\psi} \otimes \II_B)= \Tr_A (\rho \Pi \ketbra{\psi}{\psi} \Pi \otimes \II_B).
\end{equation}
Now let $\Pi \ket{\psi}=r\ket{\varphi}$ with $\ket{\varphi}$ being a two-qubit state. Note that $1=\bra{\psi} (\Pi + \ketbra{\nu}{\nu}) \ket{\psi}$, so  $\abs{\braket{\nu}{\psi}}^2=1-r^2$.
The steering outcome $\Tr_A [\rho_{\varepsilon_A} (\ketbra{\psi}{\psi} \otimes \II_B) ]$ therefore can be written as
\begin{equation}
r^2 \varepsilon_A \Tr_A (\rho \ketbra{\varphi}{\varphi} \otimes \II_B) + (1-r^2) (1-\varepsilon_A) \rho_B.
\end{equation}
This is explicitly a convex combination of two steering outcomes $\varepsilon_A \Tr_A (\rho \ketbra{\varphi}{\varphi} \otimes \II_B)$ and $(1-\varepsilon_A) \rho_B$. Trivially, the latter steering outcome $(1-\varepsilon_A) \rho_B$ can be locally simulated. We therefore need only to show that the steering outcome $\varepsilon_A \Tr_A (\rho \ketbra{\varphi}{\varphi} \otimes \II_B)$ can be locally simulated. However, this is exactly one steering outcome of the steering ensemble made by the restricted PVM $(\ketbra{\varphi}{\varphi},\ketbra{\varphi_\perp}{\varphi_\perp},\ketbra{\nu}{\nu})$, where $\ket{\varphi_\perp}$ is the two-qubit state that  is orthogonal to the two-qubit state $\ket{\varphi}$ (and the vacuum $\ket{\nu}$). The latter can be locally simulated with a LHS ensemble by assumption. Therefore, $\varepsilon_A \Tr_A (\rho \ketbra{\varphi}{\varphi} \otimes \II_B)$ can indeed be locally simulated with the LHS ensemble and $\rho_{\varepsilon_A}$ is nonsteerable with $2$-POVMs with a rank-$1$ projection component. 
\end{proof}

Now if the two-qubit state with loss $\rho_{\varepsilon_A}$ in Eq.\ \eqref{eq:two_qubit_loss} is nonsteerable with restricted PVMs, this lemma guarantees that it is also nonsteerable with $2$-POVMs with one rank-$1$ projection component. Applying Lemma~\ref{lem:noisy_POVM} with $\sigma_A=\ketbra{\nu}{\nu}$, we find that 
\begin{equation}
\frac{\varepsilon_A}{3} \rho + \left( 1-\frac{\varepsilon_A}{3} \right) \ketbra{\nu}{\nu} \otimes \rho_B
\end{equation}
is nonsteerable with arbitrary POVMs. Now retrieving the criterion for nonsteerability with restricted PVMs~\eqref{eq:restricted_PVMs}, we find the {\blu sufficient} condition for a two-qubit state with loss $\rho_{\varepsilon_A}$ to be nonsteerable with arbitrary POVMs to be
\begin{equation}
\max_{\hat{\bf{x}}} \left[ (1-3\varepsilon_A) \abs{\bf{a} \cdot \hat{\bf{x}}} + \frac{3\varepsilon_A}{2} (1+\abs{\bf{a} \cdot \hat{\bf{x}}}^2) + \norm{T\hat{\bf{x}}}\right] \le 1.
\label{eq:POVMs}
\end{equation}
This is exactly the criterion {\blu of Eqs.} (1) and (2) in the main text, except that there, we considered steering from Bob to Alice, and as a consequence Alice's Bloch vector $\bf{a}$ {\blu and heralding efficiency $\varepsilon_A$} were replaced by Bob's Bloch vector $\bf{b}$ {\blu and heralding efficiency $\varepsilon_B$}.
}

\section{II. Density matrices}

The density matrices that make up the data for Table I are displayed in Fig.\ \ref{figS1}.

\begin{figure}[thb!]
	\includegraphics[width=1\textwidth]{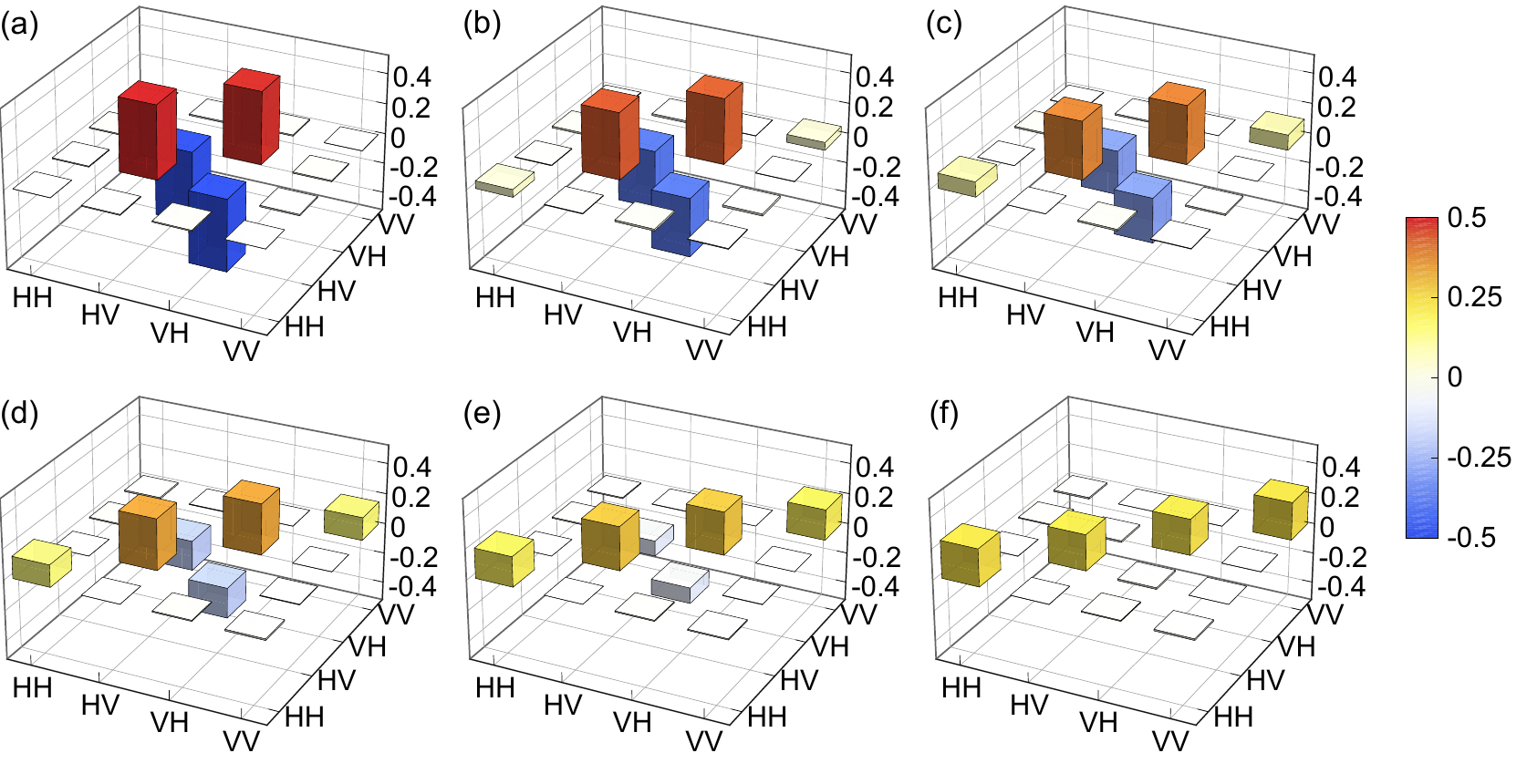}

	\caption{Experimental quantum states from the Werner-state source that are used for Table I. The plots depict the real parts of the density matrices, obtained through quantum state tomography, in order of decreasing effective $\mu$ value. The absolute values of the imaginary parts are all below 0.02.}
	\label{figS1}
\end{figure}

\section{III. Experimental Details}

\paragraph{Details of optical elements in the setup of Fig.\ \ref{fig2}.---}

The ppKTP crystal has a poling period of 46.20 $\mu \mathrm{m}$. The focal lengths of the lenses are $\mathrm{f_1}=75~ \mathrm{cm}$ and $\mathrm{f_2}=25 ~\mathrm{cm}$. The BDs are made of alpha-BBO, with a displacement of $1.4~ \mathrm{mm}$ for BD1, and $3.0~ \mathrm{mm}$ for BD2. The bandpass filter is centered at 1550 nm and has a full width at half maximum (FWHM) of 8.8 nm.

\paragraph{Heralding efficiencies.---}

The pump and detection beam waists are approximately $ 300~\mu\mathrm{m}$ and $115~\mu\mathrm{m}$, respectively. If the aim was to achieve a fixed value of $\mu\ne 0.5$, it would generally be possible to further increase the heralding efficiency in the following way: Using beam splitters with a splitting ratio other than 50:50, the output of one of the individual sources (singlet state source or maximally mixed state source) would suffer more loss, while the output of the other one would undergo less loss. This asymmetry could be exploited to lower the loss for the singlet state or the maximally mixed state, whichever has the heavier weight in the Werner state. For the extreme cases where $\mu=0$ or $\mu=1$, this approach would lead to the use of fully reflective or transmissive ``beam splitters'', which is intuitive since only one of the individual sources would be required. However, using splitting ratios other than 50:50, the overall heralding efficiency would depend on the value of $\mu$ because the heralding efficiencies from the two individual sources would not be equal, and $\mu$ determines the contribution from each. This means that the improved overall heralding efficiency at a given $\mu$ value would be achieved at the expense of a reduced tunability of $\mu$. We choose a 50:50 splitting ratio to maintain tunability over the full range of parameter values with fixed heralding efficiencies. Indeed, Figs. \ref{fig3}(a) and (b) show that $\varepsilon_A$ is nearly constant over the whole range of $\mu$ values.

\paragraph{Gating.---}

To reduce the effect of dark counts, we gate the detection of photons (singles and coincidences) based on emission events of the pump laser, by detecting in a 3 ns window around a synchronization signal from the laser. This results in average dark count values (per second) of 99 and 156 for the detectors in Alice's arm, and 55 and 113 for the detectors in Bob's arm.

\section{IV. Steering parameter}

The steering parameter for the number of measurement settings $n=6$ is defined as $S=\frac{1}{6} \sum_{k=1}^6 \langle a_k \hat{\sigma}_k^B \rangle$. Here, $a_k \in \{ -1,1 \}$ is the measurement outcome announced by Alice for those measurements where Bob has chosen the the measurement setting corresponding to the Pauli observable  $\hat{\sigma}_k^B$ in the direction $\mathbf{u}_k$. The measurement setting is chosen out of a predetermined set of six options \cite{Bennet2012}. The steering bound is a nontrivial function of $\varepsilon_A$ and is based on Alice's optimal cheating strategy \cite{Bennet2012}.


\section{V. Comparison with result from Ref.\ \cite{Wollmann2016}}

In Ref.\ \cite{Wollmann2016}, the conclusion of one-way steering was reached based on an analysis for Werner states. To this end, the Werner states closest to the experimentally obtained density matrices were used. Enabling an analysis of general states,  the work of Ref.\ \cite{Baker2018} has since indicated that even small deviations from a Werner state can have a significant bearing on proving the nonsteerability of quantum states. {\blu Here, we use our new condition for nonsteerability (Eq.\ (\ref{Traveqn})) to compare one of our states with the quantum state from Ref.\ \cite{Wollmann2016} that was thought to be one-way steerable for POVMs. We show that in contrast to our quantum state, the state from Ref.\ \cite{Wollmann2016} does not conclusively meet our condition for nonsteerability, and reveal the key experimental improvements in the present work.}

\begin{figure}[h!]
	\includegraphics[width=0.7\textwidth]{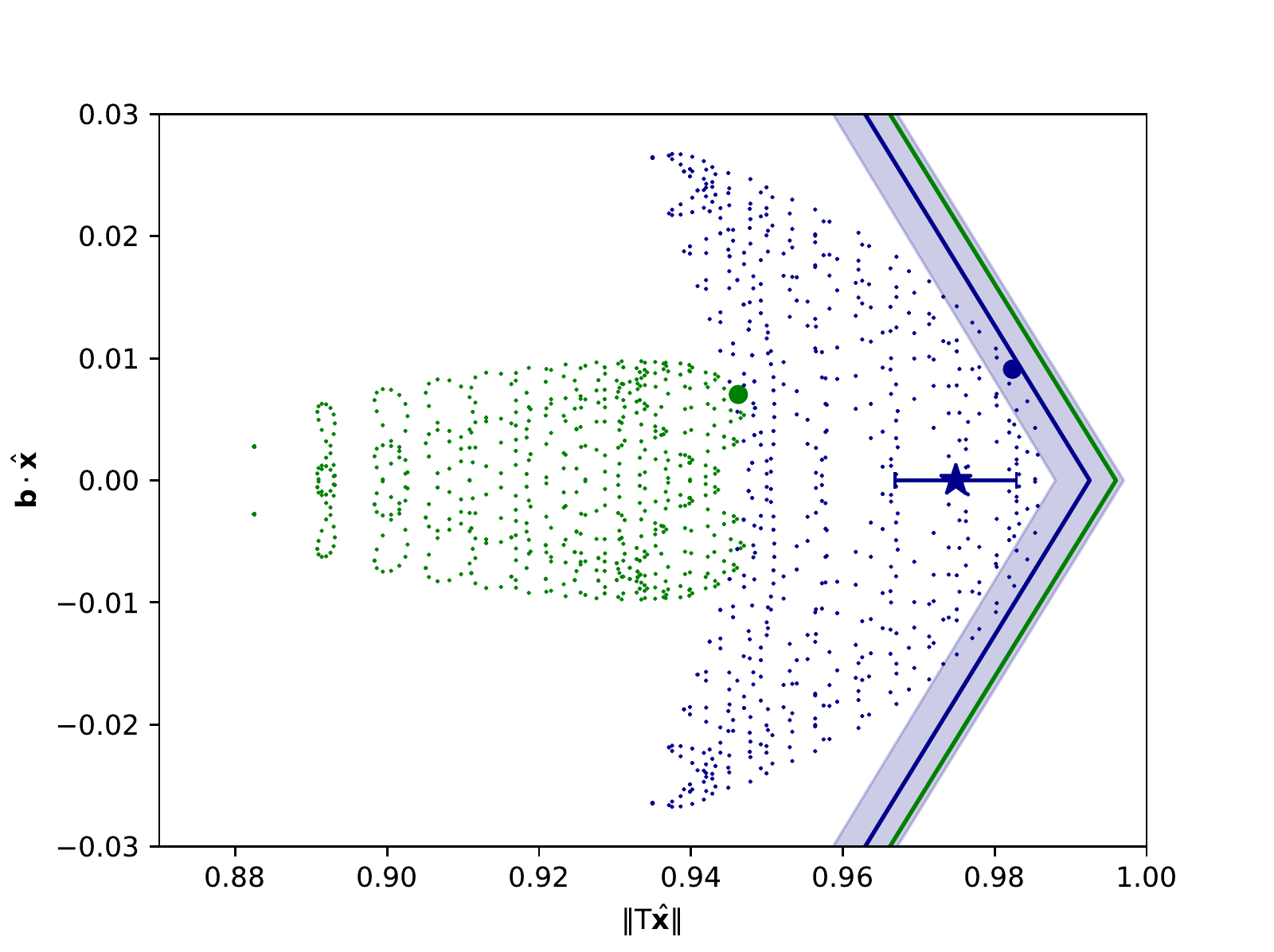}

	\caption{Comparison of the quantum state of Ref.\ \cite{Wollmann2016} (in blue) with one of our one-way steerable states (in green), with regard to the sufficient condition for nonsteerability from Bob to Alice of {\blu Eq.\ (\ref{Traveqn})}. The lines represent the sufficient condition for nonsteerability, given the experimental $\varepsilon_B$ values of $0.005\pm 0.003$ (blue) and $0.00269\pm 0.00001$ (green).  The shaded regions, visible for the blue line but too small to be discernible for the green line, show the corresponding uncertainties, based on the uncertainties of the $\varepsilon_B$ values. Each ensemble of points is associated with one quantum state, and the sufficient condition for nonsteerability (Eq.\ (\ref{Traveqn})) is equivalent to the whole ensemble being to the left of the corresponding line. An ensemble captures relevant properties of the quantum state, with individual points corresponding to specific choices of the unit vector $\hat{{\bf{x}}}$, of which a representative sample of 625 choices is shown. The choice of $\hat{{\bf{x}}}$ that constitutes the solution to the maximization in Eq.\ (\ref{Traveqn}) is marked by the larger point symbol. For the quantum state from Ref.\ \cite{Wollmann2016}, the sufficient condition for nonsteerability is not {\blu satisfied with statistical significance (note the shaded region), whereas it is for our quantum state}. The quantum state used as the one-way steerable state example from the present experiment is the one corresponding to the left-most data points in Fig.\ \ref{fig3}(c) and (d). For the case of an exact Werner state, the ensemble would collapse to a single point in the plot. {\blu An} assumption made in Ref.\ \cite{Wollmann2016} but not in this work, namely the mapping to a Werner state, is illustrated by the star. The uncertainty in its position, based on the uncertainty in the quantum state, is indicated by the horizontal error bar.
}
	\label{figS2}
\end{figure}

A useful way to visualize the relevant properties of a quantum state is through its correlation matrix $T$ and Bob's local Bloch vector $\bf{b}$, as detailed in {\blu Sec.\ I of this Supplemental Material}. Specifically important are $||T\hat{{\bf{x}}}||$, the norm of the correlation matrix multiplied by arbitrary unit vectors $\hat{{\bf{x}}}$, and ${\bf{b}}\cdot\hat{{\bf{x}}}$, the projection of the reduced Bloch vector onto the same unit vectors, plotted in Fig.\ \ref{figS2}. If the points for all choices of $\hat{{\bf{x}}}$ lie to the left of a bound, the state is nonsteerable from Bob to Alice. The bound depends on Bob's heralding efficiency $\varepsilon_B$.
 
Fig.\ \ref{figS2} shows our quantum state and that of Ref.\ \cite{Wollmann2016}, each with its corresponding bound calculated from the measured values of $\varepsilon_B$. Our bound is slightly to the right of the bound from Ref.\ \cite{Wollmann2016} because we work with a lower value of $\varepsilon_B$, and this makes witnessing nonsteerability a little easier. However, the main differences between the two cases can be ascribed to the two ensembles of points. The ensemble corresponding to our state (i) lies at lower values of $||T\hat{{\bf{x}}}||$, and (ii) has a smaller spread in ${\bf{b}}\cdot\hat{{\bf{x}}}$. 

The lower values of $||T\hat{{\bf{x}}}||$ are due to lower correlations in the state, also evidenced by a lower effective $\mu$ value. This generally helps to show nonsteerability from Bob to Alice, while at the same time making it more challenging to demonstrate steering from Alice to Bob. We have the option to operate at these lower correlation values and still violate a detection-loophole-free steering inequality from Alice to Bob, thanks to Alice's high heralding efficiency. It should also be noted that it is the tunability of our source which gives us the freedom to move to this advantageous correlation condition.

The smaller spread in ${\bf{b}}\cdot\hat{{\bf{x}}}$ is an indication of a smaller reduced Bloch vector, which would be 0 for the ideal case of an exact Werner state. The reason that the spread in ${\bf{b}}\cdot\hat{{\bf{x}}}$ matters is because the bounds do not correspond to vertical lines in the plot.

In summary, the {\blu experimental} improvement can be attributed to several factors: our ability to tune the correlations of our experimental quantum state, Alice's higher heralding efficiency that lets us demonstrate steering from Alice to Bob with lower correlations, a better state quality in the sense of a smaller reduced Bloch vector, and Bob's lower heralding efficiency.


\end{document}